\documentclass[11pt]{iopart}
\usepackage{iopams}
\expandafter\let\csname equation*\endcsname\relax
\expandafter\let\csname endequation*\endcsname\relax
\usepackage{amsmath}
\usepackage{amsfonts}
\usepackage{amsthm}
\usepackage{dsfont}
\usepackage{graphicx}
\usepackage{cite}

\newcommand{\ud}{\mathrm{d}}
\newcommand{\ui}{\mathrm{i}}
\newcommand{\ue}{\mathrm{e}}

\newcommand{\R}{\mathds{R}}

\newcommand{\C}{\mathds{C}}

\newcommand{\cH}{{\mathcal H}}

\newcommand{\calh}{\mathbf{h}}

\providecommand{\norm}[1]{\left|\left|#1\right|\right|}

\providecommand{\abs}[1]{\lvert#1\rvert}

\newcommand{\pa}{\partial}
\newcommand{\la}{\langle}
\newcommand{\ra}{\rangle}

\renewcommand{\Im}{\operatorname{Im}}
\renewcommand{\Re}{\operatorname{Re}}

\newtheorem{thm}{Theorem}[section]
\newtheorem{lem}[thm]{Lemma}
\newtheorem{prop}[thm]{Proposition}

\begin{document}
\title[Complexified Coherent States]{Complexified coherent states and quantum evolution with non-Hermitian Hamiltonians}

\author{Eva-Maria Graefe$^1$ and Roman Schubert$^2$}
\address{${}^1$ Department of Mathematics, Imperial College London,
London SW7 2AZ, UK}
\address{${}^2$ School of Mathematics, University of Bristol, Bristol, BS8 1TW, UK}

\begin{abstract}
The complex geometry underlying the Schr\"odinger dynamics of coherent states for non-Hermitian Hamiltonians is investigated. In particular two seemingly contradictory  approaches are compared: (i) a complex WKB formalism, for which the centres of coherent states naturally evolve along complex trajectories, which leads to a class of complexified coherent states; (ii) the investigation of the dynamical equations for the real expectation values of position and momentum, for which an Ehrenfest theorem has been derived in a previous paper, yielding real but non-Hamiltonian classical dynamics on phase space for the real centres of coherent states.  
Both approaches become exact for quadratic Hamiltonians. 
The apparent contradiction is resolved building on an observation by Huber, Heller and Littlejohn, that complexified coherent states are equivalent if their centres lie on a specific complex Lagrangian manifold. A rich underlying complex symplectic geometry is unravelled. In particular a natural complex structure is identified that defines a projection from complex to real phase space, mapping complexified coherent states to their real equivalents. 
\end{abstract}
%\maketitle
\vspace{0.4cm}
%\numberwithin{equation}{section}

\section{Introduction}
We analyse the geometric structure related to complexified coherent states, that is Gaussian states with a formal complex centre. These states appear naturally in situations where the classical Hamiltonian function is complex valued, or in classically forbidden regions in the description of tunneling processes. Here we focus in particular on the quantum counterpart of complex Hamiltonians, that is the quantum time dependence generated by the Schr{\"o}dinger equation with a non-Hermitian Hamilton operator. Such operators are of interest in many areas in science, in particular in physics and  chemistry. They appear, e.g., in the description of decay processes in quantum mechanics, from early models in nuclear physics to the use of complex scaling in the computation of resonances  \cite{Moi11}. In optics they naturally appear in the study of absorbing or optical active materials \cite{ElGa07,Long08}, and in chemistry absorbing complex potentials are frequently used for numerical simulations \cite{Muga04}. From a more mathematical perspective the spectral theory of non-Hermitian operators has received renewed interest recently, in particular due to questions arising from numerical analysis, such as the concept of the pseudo spectrum \cite{Tref05}. Further, the special class of non-Hermitian PT symmetric operators has received much attention recently, since these operators often possess a purely real spectrum, and have been suggested as a generalisation for the description of closed quantum systems \cite{Bend99a,Bend02b}.

Coherent states play a crucial role in the correspondence of quantum and classical systems. They can be used to quantise classical systems, but also the semiclassical limit of a quantum system can be conveniently expressed with the help of coherent states \cite{Ali00,Hep74,Hel75,Yaff82,Lit86,HubHelLit88}. In the usual formulation quantum to classical correspondence will associate with a non-Hermitian Hamilton operator a Hamilton function which is complex valued, and therefore the corresponding Hamiltonian dynamics will generate complex trajectories in a complexified phase space \cite{Xavi96,Bend99a,Dav99,Kaus00,Bend07b,Curt07,Most10,Brod11}. Considering the semiclassical limit in the sense of the Ehrenfest theorem for expectation values, on the other hand, leads to a classical dynamics on real phase space developed in \cite{GraeSchu11}. Here we will show that the seeming contradiction between these two formulations is related to an ambiguity in the definition of complexified coherent states. It is resolved by the identification of an equivalence class of coherent states with complex centres, \cite{HubHel87,HubHelLit88}, beautifully encoded in the concept of Lagrangian manifolds.

The plan of the paper is as follows. In the next section we summarise some of  the complex symplectic geometry which underlies   the properties of complexified coherent states. This will provide the framework for the following  section where we consider quantum and classical dynamics. We  focus on the discussion of quadratic non-Hermitian operators for which the quantum classical correspondence is exact and  which form the basis for semiclassical considerations of more general situations. We will in particular compare the structures which emerge from an extension of WKB theory to the non-Hermitian case with an extension of the Ehrenfest Theorem \cite{GraeSchu11}.

\section{Coherent states and complex structures}
It is well known that a manifold of coherent states can be interpreted as the phase space of the corresponding classical system \cite{Yaff82,Zhan90,Gnut98}, and how the symplectic structure of classical mechanics naturally arises from the geometry of coherent states. What is perhaps less appreciated, is that the coherent state manifold is further equipped with a metric and a complex structure, which is as well inherited to the classical system. As the metric structure does not appear in classical Hamiltonian equations of motion, it can be easily overlooked. This is different in the context of dissipative classical systems, where in addition to the symplectic flow of Hamiltonian dynamics, a metric gradient flow often appears. These types of dynamics are sometimes referred to as metriplectic flows \cite{Guha07}. It has recently been pointed out, how similar structures arise in the semiclassical limit of non-Hermitian quantum theories, where the metric of the classical phase space is provided by the metric on the space of coherent states \cite{Grae10,Grae10b,GraeSchu11}. 

Let us now recall  how certain classes of Gaussian coherent states endow classical phase space with a metric and  a complex structure. 
We begin with a brief review on the familiar case of real coherent states. 
Consider a  family of Gaussian states
\begin{equation}\label{eq:coh_state_r}
\psi_Z^B(x)= \frac{(\det \Im B)^{1/4}}{(\pi\hbar)^{n/4}} \ue^{\frac{\ui}{\hbar}[P\cdot(x-Q)+\frac{1}{2} (x-Q)\cdot B (x-Q)]} 
\,\, ,
\end{equation}
with $Z=(P,Q)\in \mathds{R}^n\times \mathds{R}^n$, and  $B\in M_n(\C)$, where $B$  is symmetric and has positive imaginary part, $\Im B>0$. This last  condition ensures that the state  is in $L^2(\R^n)$ and the prefactor is chosen such that the state is normalised to one.

This coherent state manifold can in the semiclassical limit be identified with the classical phase space via the centre $Z$, 
and the matrix $B$ defines a metric and a complex structure on phase space. The metric emerges in a natural way 
  in a phase space formulation of quantum mechanics, using for example the Wigner function, see \cite{Lit86} for the following. The Wigner function of the state \eqref{eq:coh_state_r} is a Gaussian centred around $Z=z$, and localised on the order of $\hbar$:
\begin{equation}\label{eq:Wigner}
W(z')=\frac{1}{(\pi\hbar)^{d}}\ue^{-\frac{1}{\hbar} (z'-Z)\cdot G(z'-Z)} \,\, ,
\end{equation}
where $z'=(p',q')$ denotes the coordinate and momentum variables, and the positive symmetric matrix $G$ is related to $B$ via 
\begin{equation}
\label{eq:defG}
G=\left(\begin{array}{cc} I & 0\\ -\Re B& I\end{array}\right)\left(\begin{array}{cc} [\Im B]^{-1} & 0\\ 0 & \Im B\end{array}\right)\left(\begin{array}{cc} I & -\Re B\\ 0 & I\end{array}\right).
\end{equation}
Hence the matrix  $G$ defines a metric on phase space. This metric $G$ has the additional property that it is symplectic, i.e., it satisfies $G\Omega G=\Omega$, where 
 $\Omega$ denotes the antisymmetric matrix 
\begin{equation}\label{eq:defOm}
\Omega=\left(\begin{array}{cc} 0 &-I_n\\ I_n & 0\end{array}\right).
\end{equation}
Since $\Omega\Omega=-I$ this implies that $-\Omega G \Omega=G^{-1}$ and using this it is easy to see that 
\begin{equation}\label{eq:defJ}
J:=-\Omega G\,\, ,
\end{equation}
defines  a complex structure on phase space, i.e., it satisfies $J^2=-I$. Recall that a general  $\Omega$-compatible complex 
structure on phase space is a symplectic matrix $J$ such that  $J^2=-I$,  and the matrix $\Omega J$  is positive definite. 
For later use we note that by \eqref{eq:defG} the complex structure $J$ can be expressed in terms of $B$ as
\begin{equation}\label{eq:J}
\begin{split}
J&=\begin{pmatrix}-\Re B [\Im B]^{-1} & \Im B+\Re B[\Im B]^{-1}\Re B\\ -[\Im B]^{-1} & [\Im B]^{-1} \Re B\end{pmatrix}\\
&=
\begin{pmatrix} -\Re B & I \\ - I & 0\end{pmatrix}\begin{pmatrix} [\Im B]^{-1} & 0 \\ 0 & \Im B\end{pmatrix}\begin{pmatrix} I & -\Re B\\ 0 & I \end{pmatrix}\,\, . 
\end{split}
\end{equation}Ä

The physical meanings of the centre $Z$ and the metric $G$ become apparent when considering expectation values and variances of physical observables $\hat A$. Let $\hat A$ be  the Weyl quantisation of a smooth classical observable $A$, then 
the expectation value and variance of $\hat A$ in the state  \eqref{eq:coh_state_r} are given by 
\begin{equation}
\frac{\la \psi, \hat A\psi\ra}{\norm{\psi}^2}=A(Z)+O(\hbar)\,\, ,\quad (\Delta\hat A)_{\psi}^2
= \frac{\hbar}{2}\nabla A(Z)\cdot G^{-1}\nabla A(Z)+O(\hbar^2)\,\, , 
\end{equation}
i.e., $Z$ is the centre of the phase space distribution of $\psi$ and $G$ determines its variance. Thus, in the limit of $\hbar\to0$ each coherent state collapses to a phase space point $Z$, and the matrix $G$ encodes a local metric at this point.

Let us now extend the previous considerations to the case that the coherent state is formally centred at a complex phase space point, i.e., we consider Gaussian coherent states on $\R^n$ similar to \eqref{eq:coh_state_r}, but with a complex centre $z=(p,q)\in \C^n\times \C^n$, 
\begin{equation}\label{eq:coh_state_c}
\psi_z^B(x)= \frac{(\det \Im B)^{1/4}}{(\pi\hbar)^{n/4}} \ue^{\frac{\ui}{\hbar}[p\cdot(x-q)+\frac{1}{2} (x-q)\cdot B (x-q)]} 
\,\, ,
\end{equation}
and  $B\in M_n(\C)$ is again a symmetric $n\times n$ matrix with $\Im B>0$. Note that while the centre is formally chosen complex, the wave function can still be viewed as a function of a real coordinate $x\in\R^n$, and the condition $\Im B>0$ guarantees that it is in $L^2(\R^n)$.
Similar states were considered previously by Huber, Heller and Littejohn, \cite{HubHel87,HubHelLit88}, and it was noted that different choices of the complex centre $z$ can lead to the same quantum state. In particular, it was found that two centres, $z$ and $z'$, define the same quantum state if  
\begin{equation}\label{eq:LB}
z-z'\in L_B:=\{(Bq,q)\, ;\, q\in \C^n\}\,\, ,
\end{equation}
where $L_B$ is a natural complex Lagrangian space associated with the state \eqref{eq:coh_state_c} which we will analyze in more detail below. We will show here that this result is closely related to the complex structure $J$ induced by $B$ and can be reformulated in terms of a natural projection from complex phase space to real phase space defined by 
\begin{equation}\label{eq:proj_J}
P_J(z):=\Re z +J \Im z\,\, ,
\end{equation}
i.e., $P_J(\Re z +\ui \Im z)=\Re z+J \Im z$, where $z\in \C^n\times C^n$ and the real and imaginary parts are taken component-wise. 

The main result of this section can now be formulated  as follows. 
%%%%%%%%%%%%%%%%%%%%%%%%%%%
%%%%%%%%%%%%%%%%%%%%%%%%%%%
\begin{thm}\label{thm:proj} Let $\psi_z^B$ be the coherent state \eqref{eq:coh_state_c} and $P_J$ be the projection \eqref{eq:proj_J} defined 
in terms of the  complex structure \eqref{eq:J}, then 
\begin{equation}\label{eq:proj_coh_c}
\psi_z^B=\ue^{\frac{\ui}{\hbar}\sigma(z,P_J(z))} \psi_{P_J(z)}^B\,\, ,
\end{equation}
where with $z=(p,q)$  and $P_J(z)=(P,Q)$ we have 
\begin{equation}
\sigma(z,P_J(z))=\frac{1}{2}(P+p)\cdot (Q-q)\,\, .
\end{equation}
Furthermore the Wigner function of this state is
\begin{equation}\label{eq:Wigner_c}
W(z')=\frac{\ue^{-2\Im \sigma(z,P_J(z))/\hbar}}{(\pi\hbar)^n}\ue^{-\frac{1}{\hbar} (z'-P_J(z))\cdot G(z'-P_J(z))} \,\, .
\end{equation}
\end{thm}
%%%%%%%%%%%%%%%%%%%%%%%
%%%%%%%%%%%%%%%%%%%%%%%
In other words, the complex "centre" $z=\Re z +\ui \Im z$ of the state \eqref{eq:coh_state_c} is projected 
to the real centre $Z=\Re z+J\Im z$. Hence the coherent state centred at $z$ is physically equivalent to the one centred at $Z=P_J(z)$. 

\begin{proof}
We can write a coherent state \eqref{eq:coh_state_c} in the form $\psi(x)=C\ue^{\frac{\ui}{\hbar}S(x)}$ 
with 
\begin{equation}\label{eq:Scomplex}
S(x)=p\cdot (x-q)+\frac{1}{2}(x-q)\cdot B(x-q)\,\, ,
\end{equation}
and some constant $C$. The crucial step is to note that this state is 
concentrated around the point where the imaginary part of $S(x)$ is minimal, 
but since  the parameter $z=(p,q)\in \C^n\times\C^n$ 
can be complex the minimum 
need not be located at $x=q$. Let us introduce $Z=(P,Q)$  by the conditions
\begin{equation}
\nabla \Im S(Q)=0\,\, \quad \text{and} \quad P=\nabla \Re S(Q)=\nabla S(Q)
\end{equation}
then $Q$ is the minimum of the imaginary part of $S$ and by expanding $S(x)$ up to second order 
around $x=Q$ we can rewrite $S(x)$ as 
\begin{equation}\label{eqS_S}
S(x)=S(Q)+P\cdot (x-Q)+\frac{1}{2} (x-Q)\cdot B(x-Q)\,\, .
\end{equation}
The complex structure will now appear if we express $Z$ in terms of $z$. 
We find $\nabla \Im S(x)=\Im [p+B(x-q)]=\Im p+\Im B(x-\Re q)-\Re B\Im q$ and
thus the condition $\nabla \Im S(Q)=0$ gives 
\begin{equation}\label{eq:Qq}
Q=\Re q+[\Im B]^{-1}\Re B\Im q-[\Im B]^{-1}\Im p.
\end{equation}
Since $\nabla \Re S(x)=\Re [p+B(x-q)] =\Re p+\Re B(x-\Re q)+\Im B\Im q$ we obtain further
\begin{equation}\label{eq:Pp}
P=\Re p - \Re B [\Im B]^{-1}\Im p +(\Re B [\Im B]^{-1}\Re B+\Im B)\Im q\,\, .
\end{equation}
These  two equations yield $Z=(P,Q)=P_J(z)$, with $J$ given by  \eqref{eq:J}, and hence 
with \eqref{eqS_S} we find 
\begin{equation}
\psi_z^B=\ue^{\ui S(Q)/\hbar}\psi_{P_J(z)}^B\,\, .
\end{equation}
It remains to compute $S(Q)=p\cdot(Q-q)+\frac{1}{2}(Q-q)\cdot B(Q-q)$. From \eqref{eq:Qq} and 
\eqref{eq:Pp} we find $B(Q-q)=P-p$ and hence 
\begin{equation}
S(Q)=\frac{1}{2}(P+p)\cdot (Q-q)=\sigma(z,P_J(z))\,\, .
\end{equation}
The form of the Wigner function \eqref{eq:Wigner_c} follows from  \eqref{eq:proj_coh_c} and \eqref{eq:Wigner}.
\end{proof}

In the remainder of this section we want to elucidate the complex symplectic geometry underlying and connecting the complex matrix $B$ in the definition of a coherent state \eqref{eq:proj_coh_c},  the complex structure 
$J$ \eqref{eq:J},  and the Lagrangian submanifold $L_B$ \eqref{eq:LB}. Obviously $J$ and $L_B$ are both defined in terms of $B$. We can further show that there are one-to-one relationships between all three of them. 

Let us recall that a linear subspace $L\subset \C^n\times\C^n$ is called Lagrangian if 
$\Omega|_L=0$ and $\dim L=n$, and positive Lagrangian if in addition the 
quadratic form 
\begin{equation}
h(z,z'):=\frac{\ui}{2} z\cdot \Omega \bar z'
\end{equation}
is positive on $L$, i.e., $h(z,z)>0$ for all $z\in L$. It is a well known result \cite{HorIII} that any positive Lagrangian subspace can be written in the form \eqref{eq:LB}: 
\begin{lem}\label{lem:LB}
The subspace $L_B=\{(Bq,q);\, q\in\C^n\}$ defined in \eqref{eq:LB} is Lagrangian if $B$ is symmetric, and  positive Lagrangian if $\Im B>0$. On the other hand, if $L\subset \C^n\times\C^n$ is a positive 
Lagrangian subspace then there exists a symmetric $B\in M_n(\C)$ with $\Im B>0$ such that 
$L=\{(Bx,x)\, \, ;\, x\in \C^d\}$\,\, .
\end{lem}
\begin{proof}
It is clear from the definition that $\dim L=n$. To check that $\Omega|_L=0$ we choose 
$z=(Bx,x)\in L$ and $z'=(Bx',x')\in L$ and find $z\cdot \Omega z'=-Bx\cdot
x'+x\cdot Bx'=x\cdot (B^T-B)x'=0$, 
since $B$ is symmetric. 
To check positivity we consider $h_L(z,z)=\ui z\cdot \Omega \bar z/2$ with $z=(Bx,x)$ which gives
\begin{equation}
h_L(z,z)=\frac{\ui}{2}[-(Bx)\cdot \bar x+x\cdot \bar B\bar
  x]=\frac{\ui}{2}x\cdot [ \bar B-B^T]\bar x
=x\cdot \Im B\bar x\geq 0\,\, .
\end{equation}

Now assume $L$ to be a positive Lagrangian subspace and consider the projection $\pi :L\to \C^n$ defined by $\pi(p,q)=q$. Then 
$\ker \pi=\{0\}$ because if $z=(p,q)\in \ker \pi$, then $q=0$ and hence $\ui z\cdot \Omega \bar z/2=0$, thus positivity of $L$ implies $z=0$. Therefore the map $\pi$ is invertible and since it leaves the $q$ component invariant 
the inverse must be of the form $\pi^{-1}(q)=(Bq,q)$ for some matrix $B$, i.e., $L=\{(Bq,q)\,\, ;\, q\in \C^n\}$. 
That $B$ is symmetric and has positive imaginary part follows now as before from the fact that $L$ is positive and Lagrangian. 
\end{proof}

This establishes the one-to-one correspondence between Lagrangian subspaces and complex symmetric matrices with positive imaginary part. Let us now relate complex structures and positive Lagrangian subspaces. 
By \eqref{eq:proj_J} and \eqref{eq:proj_coh_c}, the complex centres $z$ and $z'$ define the same state if $P_J(z-z')=0$. Hence, the set 
of equivalent complex centres is given by 
\begin{equation}\label{eq:LP}
L:=\ker P_J=\{z\in \C^n\times \C^n\, :\, P_J(z)=0\,\, \} \,\, .
\end{equation}
According to the work of Heller \textit{et. al.} \cite{HubHel87,HubHelLit88} we expect that $L=L_B$. Let us, however, first show that  $L$ is actually a positive Lagrangian manifold, and furthermore, that the set of $\Omega$-compatible complex structures is isomorphic to the set of positive Lagrangian subspaces of $\C^{n}\times\C^n$. 

\begin{lem} Let $J$ be a $\Omega$-compatible complex structure on $\R^n\times\R^n$ (see the definition after equation
\eqref{eq:defJ}) and define 
\begin{equation}
 P_J(z):=\Re z+J\Im z\,\, ,
\end{equation}
then 
\begin{equation}
L:=\ker P_J= \{z\in V^{\C}\,\, ;\, \Re z+J\Im z=0\}
\end{equation}
is a positive Lagrangian subspace. Conversely, for every positive Lagrangian subspace $L$ there exists a 
compatible complex structure $J_L$  such that $L=\ker
P_{J_L}$, i.e., 
\begin{equation}\label{eq:projker}
z\in L\, \Leftrightarrow\,\, \Re z+J_L\Im z=0\,\, .
\end{equation}
\end{lem}

\begin{proof} Note that since $J^2=-I$ the relation $ \Re z+J\Im z=0$ can be rewritten as 
\begin{equation}\label{eq:LgraphJ}
\Im z=J\Re z\,\, , 
\end{equation}
i.e., $z\in L$ means $z=(I+\ui J)\Re z$. Since $J$ is non-degenerate we clearly have $\dim_{\C} L=n$, and 
for $z,z'\in L$ we get 
\begin{equation}\label{eq:Lag_Jprop}
\begin{split}
z\cdot \Omega z'&=\Re z\cdot (I+\ui J^T)\Omega (I+\ui J)\Re z'\\
&=\Re z\cdot (\Omega -J^T\Omega J)\Re z'+\ui \Re z\cdot (J^T\Omega + \Omega J)\Re z' 
\end{split}
\end{equation}
and if  $J$ is symplectic and 
$\Omega J=G$ symmetric we get that $z\cdot \Omega z'=0$, and thus $L$ is Lagrangian. Furthermore we find 
for $z= (I+\ui J)\Re z\in L$ 
\begin{equation}\label{eq:LGpos}
\frac{\ui}{2} z\cdot \Omega \bar z=\Re z \cdot G \Re z\,\, ,
\end{equation}
hence $L$ is positive. 

On the other hand, assume $L\subset \C^n\times \C^n $ to be a positive Lagrangian subspace and consider the 
map $\Re_L : L\to \R^n\times \R^n$, defined by $\Re_L(z)=\Re z$. We claim that this map is invertible. To see this assume 
$z\in \ker \Re_L$, i.e, $\Re z=0$, then $\ui z\cdot \Omega \bar z=\ui \Im z\cdot \Omega \Im z=0$, hence $z =0$ by the positivity of $L$, so 
$\ker \Re_L=\{0\}$ and  $\Re_L$ is invertible as claimed. The inverse must be of the form 
$\Re_L^{-1}(v)=v+\ui Jv$ for a linear map $J:\R^n\times\R^n\to \R^n\times\R^n$. Then \eqref{eq:Lag_Jprop} with $\Re z =v$ shows that if 
$L$ is Lagrangian $J$ must be symplectic and $G:=\Omega J$ symmetric, and \eqref{eq:LGpos} 
 shows that $G$ must be positive. Then $J^2=\Omega G\Omega G=\Omega \Omega =-I$, therefore $J$ is a 
 compatible complex structure. 
 \end{proof}
   
In summary, we have shown that the set of complex symmetric matrices with positive imaginary part, the set of positive Lagrangian 
subspaces, and the  set of $\Omega$-compatible complex structures are all isomorphic to each other. What we have not shown 
yet is that $L_B$ is actually mapped to the complex structure \eqref{eq:J}, i.e., that 
\begin{equation}\label{eq:L=LB}
\ker P_J=L_B\,\, .
\end{equation}
Since $\dim \ker P_J=\dim L_B$ it is enough to show that $L_B\subset \ker P_J$, i.e., that 
for any $z\in L_B$ we have $\Re z+J\Im z=0$. Now any element in $L_B$ is of the form 
$z=(Bq,q)$ for some $q\in \C^n$  and a short calculation gives 
\begin{equation}
z=\bigg[\begin{pmatrix} \Re B & -\Im B \\ I & 0\end{pmatrix}   +\ui    \begin{pmatrix} \Im B & \Re B \\ 0 & I\end{pmatrix}\bigg]   \begin{pmatrix}\Re q\\ \Im q\end{pmatrix} \,\, ,
\end{equation}
 and hence 
$z\in \ker P_J$ for all $z\in L_B$ means
\begin{equation}
\begin{pmatrix} \Re B & -\Im B \\ I & 0\end{pmatrix}+J \begin{pmatrix} \Im B & \Re B \\ 0 & I\end{pmatrix}=0\,\, .
\end{equation}
Solving this equation for $J$ then gives the expression \eqref{eq:J} which we have already encountered. Hence 
\eqref{eq:L=LB} holds. 

For completeness we finally note  that the metric $G=\Omega J $ defines a K{\"a}hler structure on complex phase space which turns $P_J$ into an orthogonal projection: 

\begin{lem} Let $\calh(z,z'):=z\cdot G\bar z'-\ui z\cdot \Omega\bar z'$ be the hermitian inner product on 
$\C^n\times\C^n$ defined by $G=\Omega J$, then $P_J$ is the unique projection onto $\R^n\times\R^n$ which is 
hermitian with respect to $\calh$, i.e., $\calh(P_Jz,z')=\calh(z,P_Jz')$. 
\end{lem}

\begin{proof}
$P_J$ is a projection, so it is hermitian with respect to $\calh(z,z')$ if 
 the kernel and image are orthogonal to each other. Then it is as well uniquely determined by its image.  Since by \eqref{eq:LgraphJ} any $z\in L=\ker P_J$ is of the form $z=(I+\ui J)x$ for some 
$x\in \R^n\times \R^n$, we get for $z=(I+\ui J)\in L$ and $z'=x'\in \R^n\times \R^n=\Im P_J$ that 
$\calh(z,z')=x\cdot (I+\ui J)^t(G-\ui\Omega)x'=x\cdot [G+J^t\Omega+\ui(J^tG-\Omega)]x'$. But since $G$ is symmetric 
$G=\Omega J$ implies $G=-J^t\Omega$ and from $GJ=-\Omega$ we obtain $J^t G=\Omega$, therefore 
$\calh(z,z')=0$ for all $z\in \ker P_J$ and $z'\in \Im P_J$. 
\end{proof}

We have shown that coherent states with a complex centre are organised along Lagrangian submanifolds of physically equivalent coherent states one of which has a real centre. In what follows we shall investigate the time dependence of these structures under the evolution with non-Hermitian Hamiltonians. In particular, we will focus on the analytically solvable case of quadratic Hamiltonians, which lies at the heart of semiclassical considerations for more general systems.

%%%%%%%%%%%%%%%%%%%%%%%%%%%%%%%%%%%%%%%%%%%%%%%%%%%%%%%
%%%%%%%%%%%%%%%%%%%%%%%%%%%%%%%%%%%%%%%%%%%%%%%%%%%%%%%

\section{Schr{\"o}dinger dynamics with complex quadratic Hamiltonians}
Here we will investigate the Schr\"odinger dynamics generated by complex quadratic Hamiltonians that are given as Weyl quantisations of complex quadratic forms on  phase space. For these Hamiltonians semiclassical approximations are exact, and we restrict ourselves to these purely quadratic Hamiltonians to understand the essence of the dynamics in detail. It is straightforward to include also linear terms; here, however, we want to keep the discussion concise.

We continue to denote by $z=(p,q)\in \R^n\times \R^n$ points in phase space and set 
\begin{equation}
\cH(z)=\frac{1}{2} z\cdot Hz\,\, ,
\end{equation}
where $H\in M_{2n}(\C)$ is a complex  symmetric $2n\times 2n$ matrix and the 
quantum Hamiltonians we will consider are given by the Weyl quantisation
of quadratic functions of the form $\cH$,  
\begin{equation}\label{eq:Hamiltonian_quad_def}
\hat \cH =-\frac{\hbar^2}{2} \nabla_x \cdot
H_{pp}\nabla_x+\frac{\hbar}{\ui}x\cdot H_{qp}\nabla_x+\frac{1}{2}x\cdot
H_{qq} x-\frac{\ui\hbar}{2} \tr H_{qp}\,\, ,
\end{equation}
where $H=\begin{pmatrix} H_{pp} & H_{pq}\\ H_{qp} & H_{qq}\end{pmatrix}$.
We will in general allow the matrix $H$ to be time dependent without explicitly indicating this in the notation.  
Our aim is to study the solutions to the time dependent Schr{\"o}dinger equation
\begin{equation}\label{eq:Schrodinger}
\ui\hbar \pa_t\psi=\hat\cH \psi\,\, ,
\end{equation}
for initial states given by coherent states. Since our Hamilton operator is in general not self-adjoint the question of whether this equation has solutions in suitable function spaces is not trivial.  To illustrate the issue, consider the following simple example: If the Hamiltonian is given by $\cH(z)=\ui q^2/2$ the time evolution operator is of the form $U(t)=\ue^{\frac{t}{2\hbar} x^2}$ and taking for instance an 
initial state of the form $\psi_0(x)=\ue^{-\frac{b}{2\hbar} x^2}$ it follows that
\begin{equation}
\psi(t,x)=\ue^{\frac{t-b}{2\hbar} x^2}
\end{equation}
and hence $\psi(t,x)\notin L^2(\R)$ for $t\geq b$. 

Problems of this kind are avoided if the imaginary part of $H$ is chosen to be non-positive. For $\Im H\leq 0$ the Schr{\"o}dinger equation generates a contracting semigroup, and the quadratic case has been studied in some detail. We mention 
\cite{Hor95}, where the Weyl symbols of the time evolution operator have been constructed explicitly using complex symplectic geometry, and \cite{Exn83} for some early rigorous results on the damped harmonic oscillator. Here we will further analyse the consequences on the geometric structures we have highlighted in the previous section. In addition, the case of non-negative $\Im H$ is often of interest, in particular in the context of PT-symmetric quantum systems. Thus, we allow for general complex $H$ here, but we only consider special initial conditions for which explicit solutions can be computed, 
at least for short times.
 
We will investigate the dynamical behaviour of initially Gaussian coherent states that is generated by a Hamiltonian operator of the form \eqref{eq:Hamiltonian_quad_def}. Similar to the real valued case, the class of Gaussian coherent states, now with a complex centre, is invariant under this time evolution, as we shall see in the following. For this purpose we consider time dependent Gaussian coherent states of the form 
\begin{equation}\label{eq:coh_state_time}
\psi(t,x)=\ue^{\ui \alpha(t)} \frac{(\det \Im B(t))^{1/4}}{(\pi\hbar)^{n/4}} \ue^{\frac{\ui}{\hbar}[p(t)\cdot(x-q(t))+\frac{1}{2} (x-q(t))\cdot B(t) (x-q(t))]} =\ue^{\ui \alpha(t)}\psi^{B(t)}_{z(t)}(x)
\,\, ,
\end{equation}
where $z(t)=(p(t),q(t))\in \C^n\times \C^n$, $B(t)\in M_n(\C)$ is symmetric and has positive imaginary part, $\Im B(t)>0$, and 
$\alpha(t)\in \C$. Inserting the state \eqref{eq:coh_state_time} as an ansatz into the Schr{\"o}dinger equation \eqref{eq:Schrodinger} 
and separating terms with different powers of $(x-q)$ yields the following set of differential equations for $(p(t),q(t))$, $B(t)$ and $\alpha(t)$:
\begin{align}
-\dot p+B\dot q&=\cH_q'+B\cH_p'\label{eq:cz1}\\
-\dot B&= \cH_{qq}''+\cH_{pq}''B+B\cH_{qp}''+B\cH_{pp}''B\label{eq:cB1}\\
-\dot \alpha+\frac{\ui}{4}\tr(\dot BB^{-1}) &=-\frac{1}{\hbar}[p\cdot \dot q-\cH] -\frac{\ui}{2}[\tr \cH_{pq}''+\tr(\cH_{pp}''B)]\,\, \label{eq:calpha1},
\end{align}
where $\cH_p', \cH_{pq}'', ...$ denote derivatives of $\cH(z)$ with respect to $p$,  and $p$ and $q$, etc.. If we choose 
$p$ and $q$ to be solutions to Hamiltons equations, i.e., $\dot p=-\cH_q'$ and $\dot q=\cH_p'$, then the first equation is satisfied, and furthermore using the second equation we can simplify the third, thus arriving at the simplified system 
\begin{align}
\dot z&=\Omega H z \label{eq:cz}\\
\dot B&= -H_{qq}-H_{pq}B-BH_{qp}-BH_{pp}B\label{eq:cB}\\
\dot \alpha &=\frac{1}{\hbar}[p\cdot \dot q-\cH(z)] +\frac{\ui}{4}\tr[ H_{pp}B-H_{qq}B^{-1}]\label{eq:calpha}\,\, .
\end{align}
Here  the first equation is Hamilton's equation with a complex Hamilton function and the third equation can be integrated once the first and the second are solved. The solutions to 
the second equation can be obtained most easily using symplectic geometry which will be reviewed in what follows. 
This set of equations is a complex extension of the classical approach to coherent state propagation of Hepp, \cite{Hep74}, and Heller \cite{Hel75}, which is used and developed further in many areas (see, e.g., the review \cite{Lit86} or \cite{Rob07} for an overview of more recent mathematical developments).  

For complex $H$ equation \eqref{eq:cz} leads to complex solutions $z(t)$, even if the  
initial condition is chosen to be real, and thus we will obtain coherent states with complex centres. As discussed in the previous section a complex centre has no direct physical meaning, but using a complex structure it can be projected to a physically meaningful real centre. We will now apply the complex symplectic geometry we developed in the last section to understand the relation between the dynamics of the complex centre and its projection to real space. 

In a previous paper \cite{GraeSchu11} we concentrated on the dynamics of the Wigner function which directly yields the expectation values and hence the real centre of a state. This considerations led to a non-Hermitian version of Ehrenfest's theorem 
 with a new type of classical dynamics emerging in the semiclassical limit. We derived an evolution equation for the Wigner function, which in the case of a quadratic Hamiltonian  reduces to
\begin{equation}\label{eq:SchrWigner}
\hbar \pa_tW(t,z)=-\bigg(-\frac{\hbar^2}{4}\Delta_{\Im H}-\hbar z\cdot \Re H\Omega \nabla-2z\cdot \Im H z\bigg)W(t,z)\,\, ,
\end{equation}
where all derivatives are with respect to $z$, and 
\begin{equation}
\Delta_{\Im H}:=-\nabla\cdot \Omega^T\Im H\Omega\nabla\,\, .
\end{equation}
While the evolution equation in \cite{GraeSchu11} for general Hamiltonians is a semiclassical approximation, the quadratic case \eqref{eq:SchrWigner} is exact. 

If $\psi$ is  of the type \eqref{eq:coh_state_time} the Wigner function is of the form 
\begin{equation}\label{eq:AWigner}
W(t,z)=\frac{\ue^{-\beta(t)}}{(\pi\hbar)^n} \ue^{-\frac{1}{\hbar}(z-Z(t))\cdot G(t)(z-Z(t))}
\end{equation}
with $Z(t)\in \R^n\times\R^n$, a symmetric $G(t)\in M_{2n}(\R)$, and $\beta(t)\in\R$. 
Inserting the ansatz \eqref{eq:AWigner} into equation \eqref{eq:SchrWigner}, and separating different powers of $(z-Z)$, 
leads to the following set of equations
\begin{align}
\dot Z&=\Omega \Re H Z+G^{-1} \Im H Z\label{eq:rZ}\\ 
\dot G&=\Re H\Omega G-G\Omega \Re H-\Im H+G\Omega^T\Im H\Omega G\label{eq:rG}\\
\dot \beta&=-\frac{2}{\hbar} Z\cdot \Im H Z-\frac{1}{2} \tr[\Im H \Omega G\Omega^T]\label{eq:beta}
\end{align}
It can be verified, that this set of equations is also compatible with the dynamical equations \eqref{eq:cz1}, \eqref{eq:cB1}, and \eqref{eq:calpha1} obtained from the coherent state ansatz in the Schr{\"o}dinger equation, if we demand $p$ and $q$ to be real. Thus, equations \eqref{eq:cz}, \eqref{eq:cB}, and \eqref{eq:calpha} are not the unique dynamical equations for the propagation of coherent states for non-Hermitian Hamiltonians. 

The two different sets of equations that we have obtained, \eqref{eq:cz}, \eqref{eq:cB}, and \eqref{eq:rZ}, \eqref{eq:rG}, are supposed to describe the dynamics of the same physical state. In what follows we will discuss how they can be related using complex structure associated with the coherent states. 

\subsection{Symplectic evolution}
To solve the evolution equations obtained above, in particular the nonlinear matrix Ricatti equations \eqref{eq:cB} and \eqref{eq:rG}, we have to understand how the geometric structures discussed in the previous section evolve in time under the action of complex Hamiltonian dynamics. For this purpose, we first investigate the action of a linear symplectic map on a positive Lagrangian subspace $L$, i.e., we change $L$ to $SL$ with $S\in Sp(n, \C)$.
Here $Sp(n,\R)$ and $Sp(n,\C)$ denote the set of real or complex $2n\times 2n$  matrices $S$ with $S^T\Omega S=\Omega$, i.e., the real and complex linear symplectic groups.  
Since any $z\in SL$ is of the form $z=Sz_0$ for some $z_0\in L$ we get 
$z\cdot \Omega z'=z_0\cdot S^T\Omega S z_0'=z_0\cdot \Omega z_0'=0$, since $L$ is Lagrangian, and thus $SL$ is, too. Furthermore  
\begin{equation}
\frac{\ui}{2} z'\cdot \Omega \bar z'=\frac{\ui}{2} z\cdot S^T \Omega \overline{S} \bar z\, ,
\end{equation}
thus, if $\overline S=S$, i.e., $S\in Sp(n, \R)$, then $SL$ is positive, too. If $S$ is complex, $SL$ does not have to be positive any more. 

We will mainly consider situations in which $S$ is the solution to Hamilton's equation, i.e, $S(t)$ satisfies
\begin{equation}\label{eq:HamS}
\dot S=\Omega H S\,\, ,\quad\text{with}\quad S(t=0)=I\,\, ,
\end{equation} 
where $H\in M_{2n}(\C)$ is symmetric.

\begin{lem} Assume $L$ to be a positive Lagrangian subspace and $S(t)$ a solution of \eqref{eq:HamS}. Then there exists a $T_{H,L}$ such that  for all $t \in [0, T_{H,L}) $ 
$S(t)L$ is again a positive Lagrangian subspace. If  $\Im H\leq 0$ we can take $T_{H,L}=\infty$.
\end{lem}

\begin{proof} Since $S(t)$ is close to the identity for small $t$ ,  $S(t)L$ will be positive by continuity for sufficiently small $t$. 
If $\Im H\leq 0$ we proceed as follows. With $Sz\in SL$ for $z\in L$ we have to consider $\ui (Sz)\cdot \Omega \overline{S z}/2=\ui z\cdot S^T\Omega \bar S \bar z/2$ for $z\in L$. From  \eqref{eq:HamS} we find 
\begin{equation}
\frac{\ud}{\ud t }\bigg(\frac{\ui}{2} S^T\Omega \bar S\bigg)=\frac{\ui}{2}S^T[-H\Omega\Omega +\Omega \Omega \bar H ]\bar S=- S^T\Im H \bar S.
\end{equation}
Thus, if $\Im H\leq 0$ then $\frac{\ud}{\ud t }\frac{\ui}{2}  z\cdot S^T\Omega \bar S \bar z\geq 0$ and  $S(t)L$ is therefore positive for all 
$t\geq 0$.  
\end{proof}

Let us now investigate how $B$ and the complex structure transform if we apply a symplectic map to $L$.
\begin{prop}\label{prop:actS} Let $L$ be a positive Lagrangian subspace and $S\in Sp(n,\C)$ such that  $SL$ is still positive. Then 
\begin{itemize}
\item[(i)] 
\begin{equation}\label{eq:BSL}
B_{SL}=S_*B_L
\end{equation}
 where the action of 
$S=\begin{pmatrix} S_{pp} & S_{pq}\\ S_{qp} & S_{qq}\end{pmatrix}$ on $B_L$ is defined by \\
\begin{equation}
S_*B_L:=(S_{pp}B_L+S_{pq})(S_{qp}B_L+S_{qq})^{-1}\,\, ,
\end{equation}
\item[(ii)] and 
\begin{align}
J_{SL}&=(\Re S-\Im S J_L)J_L(\Re S-\Im S J_L)^{-1}\label{eq:JSL}\\
G_{SL}&=\Omega(\Re S-\Im S J_L)\Omega^T G_{L}(\Re S-\Im S J_L)^{-1}\label{eq:GSL}\,\, . 
\end{align}
\end{itemize}
\end{prop}

\begin{proof}
Let $z\in L$, then there exists a $q\in \C^n$ such that $z=(B_L q,q)$, by Lemma \ref{lem:LB},  and since $Sz\in SL$ there exists a $q'\in \C^n$ such that 
$Sz=(B_{SL}q',q')$. Now $Sz=S(B_L q,q)=(S_{pp}B_L q+S_{pq}q, S_{qp}B_Lq+S_{qq}q)$ and hence we obtain the two equations
\begin{equation}
(S_{pp}B_L +S_{pq})q=B_{SL}q'\,\, ,\quad (S_{qp}B_L+S_{qq})q=q'\,\, .
\end{equation}
From the second equation we get $q=(S_{qp}B_L+S_{qq})^{-1}q'$ and inserting this into the first gives 
$(S_{pp}B_L +S_{pq})(S_{qp}B_L+S_{qq})^{-1}q'=B_{SL}q'$, which is the first result. 

To derive the second result we note that $z\in L$ means $z=\Re z+\ui J_L\Re z$, by \eqref{eq:LgraphJ}, and similarly $Sz\in SL$ means 
$Sz=\Re (Sz)+\ui J_{SL} \Re (Sz)$ and thus we arrive at the expressions 
\begin{align}
\Im [Sz] &=J_{SL}\Re [Sz]=J_{SL} (\Re S-\Im S J_L)\Re z\\
\Im [Sz] &=\Im [S(\Re z+\ui J_L\Re z)]=(\Im S+\Re S J_L)\Re z\,\, .
\end{align}
Comparing these two expressions for $\Im (Sz)$ gives $J_{SL}= (\Im S+\Re S J_L)(\Re S-\Im S J_L)^{-1}$ and 
with $J_L^2=-1$ we furthermore obtain $(\Im S+\Re S J_L)=(\Re S-\Im S J_L)J_L$. The result for $G_{SL}$ then follows from $G_{SL}=-\Omega J_{SL}$. 
\end{proof}

Note that there is a certain similarity in the structure of equations 
\eqref{eq:BSL}, \eqref{eq:JSL} and \eqref{eq:GSL}. In fact we can rewrite 
 \eqref{eq:JSL} and \eqref{eq:GSL} as 
 \begin{align}
 J_{SL}&=(\Re S J_L+\Im S)(-\Im S J_L +\Re S)^{-1}\\
 G_{SL}&=(\Omega \Re S\Omega^T G_L+ \Omega \Im S)( -\Im S \Omega^T G_L+ \Re S)^{-1}
 \end{align}
then 
\begin{equation}
J_{SL}=\tilde \Phi_*J_L\,\, \quad \text{and} \quad G_{SL}=\Phi_*G_L\,\, ,
\end{equation}
with 
\begin{equation}\label{eq:phis}
\tilde \Phi=\begin{pmatrix} \Re S & \Im S\\ -\Im S & \Re S\end{pmatrix} \,\, \quad \text{and}\quad  
\Phi= \begin{pmatrix} \Omega \Re S\Omega^T &  \Omega \Im S\\ -\Im S \Omega^T & \Re S\end{pmatrix}\,\, .
\end{equation}

If the symplectic matrix $S$ is a solution of the differential equation \eqref{eq:HamS} then this induces corresponding differential equations for the evolution of the  matrices
$B_{SL}$ and $J_{SL}$ which we shall now derive. 

\begin{thm} \label{thm:S-prop} Let $S(t)$ be a solution to \eqref{eq:HamS} with $H=\begin{pmatrix} H_{pp} & H_{pq}\\ H_{qp} & H_{qq}\end{pmatrix}$, 
and $L$ a positive Lagrangian subspace, then there exists a $T_{H,L}>0$ such that $S(t)L$ is positive  for $t\in [0, T_{H,L}]$ and 
we have 
\begin{itemize}
\item[(i)] 
\begin{equation}\label{eq:Bdot}
\dot B_{SL}= -H_{qp}B_{SL}-B_{SL}H_{pq}-H_{qq}-B_{SL} H_{pp}B_{SL}
\end{equation}
%%
%5
\item[(ii)]  and 
\begin{align}
\dot J_{SL}&=  \Omega \Re H J_{SL}-J_{SL}\Omega \Re H+\Omega \Im H+J_{SL} \Omega \Im HJ_{SL}\\
\dot G_{SL} &=  \Re H \Omega G_{SL}-G_{SL}\Omega \Re H- \Im H+G_{SL} \Omega^T \Im H\Omega G_{SL}\label{eq:Gdot}
\end{align}
\end{itemize}
Furthermore if $\Im H \leq 0$ we can take $T_{H,L}=\infty$. 
\end{thm}

\begin{proof}
Since $S(0)=I$ it is clear that for small $t$ the space  $S(t)L$ will still be positive, hence there exists a $T_{H,L}$ such that 
$SL$ is positive for $t\in [0,T_{H,L}]$.  Now from   \eqref{eq:HamS}  we get 
\begin{equation}\label{eq:dotS}
\begin{pmatrix} \dot S_{pp} & \dot S_{pq}\\ \dot S_{qp} & \dot S_{qq}\end{pmatrix} 
=\begin{pmatrix} -H_{qp}S_{pp}-H_{qq}S_{qp} & -H_{qp}S_{pq}-H_{qq}S_{qq}\\ H_{pp} S_{pp}+H_{pq}S_{qp} & H_{pp}S_{pq}+H_{pq}S_{qq}\end{pmatrix}
\end{equation}
then differentiating the relation \eqref{eq:BSL} and using \eqref{eq:dotS} gives 
\begin{equation}
\begin{split}
\dot B_{SL}
&=(\dot S_{pp}B_L+\dot S_{pq})(S_{qp}B_L+S_{qq})^{-1}-B_{SL}(\dot S_{qp}B_L+\dot S_{qq})(S_{qp}B_L+S_{qq})^{-1}\\
&=-H_{qp} (S_{pp}B_L+S_{pq})(S_{qp}B_L+S_{qq})^{-1} -H_{qq} (S_{qp}B_L+S_{qq})(S_{qp}B_L+S_{qq})^{-1} \\
&\quad -B_{SL}H_{pp}(S_{pp}B_L+S_{pq})(S_{qp}B_L+S_{qq})^{-1} -B_{SL}H_{pq}(S_{qp}B_L+S_{qq})(S_{qp}B_L+S_{qq})^{-1}\\
&=-H_{qp}B_{SL}-H_{qq}-B_{SL} H_{pp}B_{SL}-B_{SL}H_{pq}
\end{split}
\end{equation}
To prove the second set of relations we first rewrite \eqref{eq:JSL} as $J_{SL}=AJ_LA^{-1}$ with 
$A=\Re S-\Im SJ_L$, then 
\begin{equation}
\dot J_{SL}=\dot A J_{L}A^{-1}-AJ_L A^{-1}\dot A A^{-1}=\dot A A^{-1}J_{SL}-J_{SL}\dot A A^{-1}\,\, . 
\end{equation}
Then from \eqref{eq:HamS} we find $\Re \dot S=\Omega \Re H\Re S-\Omega \Im H\Im S$ and 
$\Im \dot S=\Omega \Im H\Re S+\Omega \Re H \Im  S$ and using these relations we find  
\begin{equation}
\dot A=\Re \dot S-\Im \dot S J_L=[\Omega \Re H -\Omega \Im H J_{SL}]A
\end{equation}
and this leads to 
\begin{equation}
\dot J_{SL}=\Omega \Re H J_{SL}+\Omega \Im H-J_{SL}\Omega \Re H+J_{SL} \Omega \Im HJ_{SL}\,\, .
\end{equation}
The result for $G_{SL}$ then follows using the relation $G_{SL}=\Omega J_{SL}$
\end{proof}

The formal similarity of the equations \eqref{eq:Gdot} and \eqref{eq:Bdot} suggests to define 
a Hamiltonian $K(\zeta,z)$ on the doubled phase space  by 
\begin{equation}
K(\zeta,z)=\frac{1}{2} (\zeta,z)\begin{pmatrix} -\Omega^T\Im H\Omega & \Omega\Re H \\ -\Re H \Omega & \Im H\end{pmatrix}
\begin{pmatrix}\zeta \\ z\end{pmatrix}
\end{equation}
then the matrix $\Phi(t)$ from \eqref{eq:phis} satisfies 
\begin{equation}\label{eq:phi-t}
\dot\Phi =\begin{pmatrix} 0 & -I\\ I & 0\end{pmatrix} \begin{pmatrix} -\Omega^T\Im H\Omega & \Omega\Re H \\ -\Re H \Omega & \Im H\end{pmatrix} \Phi\,\, .
\end{equation}
And so by solving \eqref{eq:phi-t} with $\Phi(t=0)=I$ we find a matrix such that 
\begin{equation}\label{eq:Gstar}
G(t)=\Phi(t)_*G
\end{equation}
is a solution to \eqref{eq:rG} with $G(t=0)=G$.

\subsection{Quantum evolution}

The results from Theorem \ref{thm:S-prop} allow us to solve the non-linear Riccati equations 
\eqref{eq:cB} and \eqref{eq:rG} in terms of solutions to  linear Hamiltonian equations, which we will exploit in what follwos. 

We first consider the Schr{\"o}dinger equation for a coherent state in position  representation, \eqref{eq:coh_state_time}. 
Let $S(t)\in Sp(n,\C)$ be the solutions to 
\begin{equation}
\dot S=\Omega HS\,\, ,\quad \text{with}\quad S(0)=I
\end{equation}
then $z(t)=S(t)z_0$ is a solution to \eqref{eq:cz} and by Theorem  \ref{thm:S-prop}, part (i),  and 
Proposition \ref{prop:actS}, part (i), $S_*B$ is a solution to \eqref{eq:cB}. Hence we conclude

\begin{thm} Let $L=L_B$ be a positive Lagrangian subspace, then there exists a $T_{H,L}>0$ such that for 
$t\in [0,T_{H,L})$ the solution to the Schr{\"o}dinger equation with $\psi(t=0)=\psi_z^B$ is given by 
\begin{equation}
\psi(t)=\ue^{\ui\alpha(t)}\psi_{S(t)z}^{S(t)_*B}
\end{equation}
where $\alpha\in \C$ is the solution to \eqref{eq:calpha} with $\alpha(0)=0$. 
\end{thm}
The phase factor is related to the action along $S(t)z$ and also contains Maslov-phase type contributions. 

The matrix $S(t)_*B$ defines a time dependent complex structure $J(t)$ via \eqref{eq:J} which projects the 
complex centre $S(t)z$ to the real centre 
\begin{equation}
Z(t)=P_{J(t)}(S(t)z)
\end{equation}
and we can use Theorem \ref{thm:proj} to express the Wigner function in terms of projections from the 
complex dynamics $S(t)$. 
 
Alternatively we can solve the purely real set of equations \eqref{eq:rG} and \eqref{eq:rZ} to directly obtain the motion of the real centre. Let $\Phi(t)$ be the solution to \eqref{eq:phi-t} and $Z(t)$  a solution to \eqref{eq:rZ} with 
$G(t)=\Phi(t)_*G$ then we have 
\begin{thm} Let $G$ be a symplectic positive definite symmetric matrix. Then there exists a $T_{H,G}>0$ such that 
for $t\in [0, T_{H,G})$ the unique solution to the Wigner von Neuman equation \eqref{eq:SchrWigner} 
with initial condition $W(z)=\frac{1}{(\pi\hbar)^n}\ue^{-\frac{1}{\hbar} (z-Z)\cdot G(z-Z)}$ is given by 
\begin{equation}
W(t,z)=\frac{\ue^{-\beta(t)}}{(\pi\hbar)^n}\ue^{-\frac{1}{\hbar} (z-Z(t))\cdot [\Phi(t)_*G](z-Z(t))}
\end{equation}
where $\beta(t)\in \R$ is a solution to \eqref{eq:beta} with $\beta(0)=0$. 
\end{thm}

One of the characteristic features of the dynamical equation \eqref{eq:rZ} for the real centre $Z(t)$ is that it is in general not autonomous, the coefficients of this equation will depend on $t$ via the metric $G(t)$. However, in many cases there are special solutions for which the metric is time independent, corresponding to fixed points of the evolution equation \eqref{eq:rG}. To analyse the possible time independent complex structures we have to set the expression for the time derivative of the metric $G(t)$ in \eqref{eq:rG}, or equivalently the time derivative of the matrix $B(t)$ in \eqref{eq:cB}, to zero. Thus we obtain quadratic matrix equations for the fixed points $G_0$ and $B_0$, respectively. 
Let us illustrate this observation with a few examples. 
\begin{itemize}
\item[(1)] Assume the Hamiltonian is anti-Hermitian, i.e., $\Re H=0$, then \eqref{eq:rG} with $\dot G=0$ becomes 
$\Im H=G\Omega^T\Im H \Omega G$, and if we assume furthermore that for some $\gamma>0$ we have $\Im H=-\gamma S$,  where 
$S$ is symplectic, symmetric and positive, then we find that $G=S$ (here we used that $\Omega^TS\Omega=S^{-1}$). 
The assumptions on $\Im H$ hold for instance if $n=1$ and $\Im H$ is negative definite (with $\gamma=\det \Im H$). Thus in this case the metric and the 
associated complex structure are constant and the equation of motion for the centre simplifies to
\begin{equation}
\dot Z=-S^{-1} \gamma S Z=-\gamma Z\,\, .
\end{equation}
Hence we find a uniform contraction towards the origin. For this Hamiltonian we can as well solve   \eqref{eq:phi-t} 
explicitly and using \eqref{eq:Gstar} we  find that for an arbitrary initial $G_0$ the solution to  \eqref{eq:rG} is given by
\begin{equation}
G(t)=(G_0+\tanh(\gamma t) S)(\tanh(\gamma t)G_0+S)^{-1} S=S+O(\ue^{-\gamma t})\,\, ,
\end{equation}
hence the stationary solution $G=S$ we found above  is a global attractor, to which any other solution converges exponentially fast to. 
Note as well that $G(t)$ can be  extended to some negative $t$ but will eventually become singular. 
  
\item[(2)] We previously discussed the example $\cH(z)=\ui q^2/2$ as a case where the time evolution will be only defined
for finite time. For the discussion of this case it is most convenient to use \eqref{eq:cB}, which gives $\dot B=-\ui $ and hence 
$B(t)=B_0-\ui tI$. Since $\Im B(t)=\Im B_0-t I$ we see that the condition $\Im B(t)>0$ holds only for a finite time, after which 
the corresponding metric will blow up. The equation for the centre can also easily be solved; assume for simplicity that 
$\Re B_0=0$, then $\dot P=0$ and the position reaches infinity in finite time  
\begin{equation}
Q(t)=\frac{B_0}{B_0-t}Q_0\,\, .
\end{equation}
\item[(3)] We now have a look at a harmonic oscillator with damping induced by a momentum dependent imaginary part. 
We choose 
\begin{equation}
\cH(p,q)=\frac{\bar \delta^2}{2} p^2+\frac{\omega^2}{2}q^2\,\, , 
\end{equation}
where the parameter $\delta\in\C$ is assumed to satisfy $\abs{\delta}=1$ and $\Re \delta, \Im \delta >0$, hence $\Im \cH(p,q)=-\Re \delta \Im \delta\,  p^2\leq 0$. Therefore $\delta$ parametrizes the strength of the damping relative to the 
kinetic energy. Note that choosing $\abs{\delta}\neq 1$ just amounts 
 to rescaling of $\omega\to\abs{\delta}\omega$ and $t\to \abs{\delta}t$. 
Using \eqref{eq:cB} we find that $B=\ui \omega \delta$ is a constant solution with $\Im B>0$. We can then determine the corresponding metric $G$ and the equations of motions for the centre which read
\begin{equation}
\dot p=-\omega^2 q-2\omega \Im \delta \, p\,\, ,\quad \dot q=p\,\, .
\end{equation}
For comparison  with the classical damped harmonic oscillator we transform this set of first order equations into a second order equation 
for $q$, 
\begin{equation}
\ddot q+2\omega \Im \delta\,\,  \dot q+\omega^2 q=0\,\, .
\end{equation}
We see that due to the metric this describes an underdamped oscillator, since $\Im \delta\leq \abs{\delta}=1$, 
 irrespective of the choice for $\delta$. 
\item[(4)] It is instructive to include an example with a linear term in $z=(p,q)\in\R^2$, 
\begin{equation}
\cH_{\gamma}(z)=\frac{1}{2} z\cdot z+\ui \gamma\cdot \Omega z
\end{equation}
where $\gamma\in \R^2$. The inclusion of $\Omega$ in  the linear term is convenient, it  implies that 
if $z$ is to the right of $\gamma$ the term is negative and we have damping, and if 
$z$ is to the left of $\gamma$ the term is positive and  we have enhancement. This is a 
PT symmetric system. It can be brought to the more familiar form ${\mathcal H(p,q)}=\frac{1}{2}p^2+V(q)$ with $V(-q)=\bar V(q)$ via a canonical rotation of the phase space variables. Since $\Re H=I$ and $\Im H=0$ we find that 
$G=I$ is a solution for all times and with this initial choice the equation of motion for the centre $Z$, see \cite{GraeSchu11}, becomes
$\dot Z=\Omega(Z-\gamma)$. Hence $Z(t)=\gamma+O(t)(Z_0-\gamma)$ with 
$O(t)\in SO(2)$ denoting a rotation by $t$.  We can as well solve the equation for $\beta$ and find
$\beta(t)=-\gamma\cdot(Z(t)-Z_0)$. Thus, the centre of the Wigner function evolves along circles as for the harmonic oscillator, but the circles are shifted due to damping and enhancement in different parts of phase space. The relation to a real harmonic oscillator ${\mathcal H_0}$ can be directly seen in the following way. Introducting the complex translation $T(\gamma):=\ue^{\frac{-1}{\hbar}\gamma\cdot \hat z}$ we have $T(\gamma)^{-1}\hat\cH_0T(\gamma)=\cH_{\gamma}-\abs{\gamma}^2/2$, and thus the operator is conjugated to the harmonic oscillator by a non-unitary operator, and thus the spectrum is purely real. The norm of the state stays bounded over time although it oscillates, which reflects the fact that the eigenvalues of the Hamiltonian are real, but the eigenfunctions are not orthogonal for $\gamma\neq 0$.  

\end{itemize}

\section{Summary}
Coherent states are a useful tool for the investigation of semiclassical limits of quantum theories. The investigations presented here can be viewed as part of a programme to understand the classical dynamics emerging from the semiclassical limit 
of general non-Hermitian operators. We recently formulated an Ehrenfest Theorem for 
non-Hermitian operators \cite{GraeSchu11}, in which the classical dynamics is given by a combination of a symplectic and a metric gradient field, which are generated by the real and imaginary part of the Hamilton function, 
respectively.   This is a very different type of dynamics compared to what one would expect from extending standard WKB theory to complex Hamiltonians, which results in a Hamiltonian flow on complexified phase space.  
The main result here, is the proof that these two approaches are physically equivalent and are related by a projection from complexified phase space to real phase space,  
\begin{equation}
\ui\mapsto J\,\, ,
\end{equation}
where $J$ is a complex structure on phase space which is determined by the physical states and becomes a dynamical 
variable in our theory. 

We restricted ourselves to quadratic Hamiltonians and Gausssian coherent states here, because both 
semiclassical approaches become exact in this case, and we could focus on the complex symplectic geometry relating them. 
This will form the basis for extension to more general systems following \cite{GraeSchu11}. It is well known that semiclassical 
methods based on dynamics in complexified phase space  often run into difficulties related to analytic extensions, 
e.g., complex trajectories often develop singularities, and Hamiltonians which are very close on real phase space 
can have very different analytic extensions. The results described in the present paper and in \cite{GraeSchu11} provide an alternative approach which is non-Hamiltonian but manifestly real and thus avoids these problems. 

\ack EMG acknowledges support from the Imperial College JRF scheme.

\end{document}